\documentclass[lettersize,journal]{IEEEtran}
\usepackage{amsmath,amsfonts}

\usepackage{algorithm}
\usepackage{algorithmic}
\usepackage{amsthm} 
\usepackage{graphicx}
\usepackage{subcaption}
\usepackage{amsmath}
\usepackage{dsfont}

\def\BibTeX{{\rm B\kern-.05em{\sc i\kern-.025em b}\kern-.08em
    T\kern-.1667em\lower.7ex\hbox{E}\kern-.125emX}}
\usepackage{balance}

\begin{document}

\title{Design of Adaptive Hybrid Downlink NOMA-TDMA for Visible Light Communications Networks}

\author{Tuan A. Hoang, Chuyen T. Nguyen, and Thanh V. Pham, ~\IEEEmembership{Senior Member,~IEEE}

\thanks{Tuan A. Hoang and Chuyen T. Nguyen are with the School of Electrical and Electronic Engineering, Hanoi University of Science and Technology, Hanoi, Vietnam (email: tuan.ha203826@sis.hust.edu.vn, chuyen.nguyenthanh@hust.edu.vn).

Thanh V. Pham with the Department of Mathematical and Systems Engineering, Shizuoka University, Shizuoka, Japan (e-mail: pham.van.thanh@shizuoka.ac.jp). 
}}

\maketitle

\begin{abstract}
This paper proposes an adaptive hybrid non-orthogonal multiple access (NOMA)-time division multiple access (TDMA) scheme for multi-user visible light communication (VLC) networks, aiming to enhance users' sum-rate performance while maintaining low complexity.  In the proposed scheme, users are divided into groups where each group is served in a different time slot using TDMA. Within each group, up to two users can be served simultaneously using NOMA. A central challenge lies in determining which users should be paired together for NOMA, as the effectiveness of successive interference cancellation (SIC) employed by NOMA depends on the difference between users’ channel gains. 
To address this, for a pair of users, we determine the range of their channel gain ratio within which the pair benefits more from NOMA or TDMA. Identifying the lower and upper bounds of this range is formulated as two optimization problems which are solved efficiently using the Successive Convex Approximation (SCA) method. Simulation results demonstrate that the proposed scheme outperforms the conventional hybrid NOMA-TDMA method under different numbers of users and transmit LED powers.
\end{abstract}



\section{Introduction}
Driven by the widespread use of wireless devices and data-intensive applications such as cloud computing, video streaming, online gaming, and virtual reality (VR), the demand for reliable, high-speed, and low-latency wireless communication has grown exponentially in recent years \cite{Tariq2020}. As the limited frequency resources of radio frequency (RF) systems are becoming insufficient to fulfill this growth, visible light communications (VLC) has emerged as a promising complementary technology to RF for future wireless communications \cite{ariyanti2020visible}. Experiments have shown that VLC could offer significantly higher data rates without causing interference with existing RF systems \cite{tavakkolnia2018energy,miranda2023review}.
Furthermore, VLC can be easily integrated with existing lighting infrastructure, making it a cost-effective and environmentally friendly technology. 

To share the available communication resources among multiple users, VLC systems can adopt multiple access techniques traditionally used in RF systems, such as frequency division multiple access (FDMA), or time division multiple access (TDMA). These orthogonal multiple access (OMA) schemes allocate distinct frequency or time resources to each user to avoid interference.
However, OMA schemes often suffer from low spectral efficiency, especially when the number of users increases.
To overcome this limitation, non-orthogonal multiple access (NOMA) has gained significant attention as a key enabling technology for future communication systems, including VLC \cite{sadat2022survey,ding2017survey}. 
In contrast to OMA, NOMA allows multiple users to access the same frequency and time resources simultaneously. The most prevalent form of NOMA is power domain NOMA, in which users are assigned different power levels based on their channel conditions: users with stronger channel gains are assigned lower power, while those with weaker channel gains are allocated higher power \cite{islam2016power}. 
At the receiver, successive interference cancellation (SIC) is used to decode the signals in order of their power levels. The signal with the higher power, typically from the user with the weak channel gain, is decoded first, treating the strong user's signal as interference. Once successfully decoded, the weak user’s signal is subtracted from the received signal, allowing the receiver to decode the strong user’s signal.
Despite its advantage in improving spectral efficiency, NOMA introduces additional complexity due to the requirement of SIC, which can increase both decoding delay and computational load at the receiver, particularly when the number of users is large. To balance the trade-off between performance and complexity, many practical implementations adopt a hybrid OMA-NOMA approach, typically based on a two-user model \cite{ding2015impact, zhu2018optimal}. In this configuration, two users are paired to share the same resource block using NOMA, and different NOMA pairs are assigned to separate resource blocks, such as different time slots in the case of TDMA. 

In this study, we propose an adaptive hybrid NOMA-TDMA system, where users are divided into groups, each allocated a distinct time slot. Unlike the approaches in \cite{ding2015impact, zhu2018optimal}, which constrain each group to consist of exactly two users, our approach allows each group to contain either one or two users, depending on the ratio of their channel gains. This is motivated by the fact that, in typical indoor VLC environments, the dominant line-of-sight transmission and the relatively small coverage area of VLC cells lead to users experiencing similar channel conditions (i.e., high channel correlation). In such cases, the power level separation required for effective SIC becomes less pronounced, which can degrade the performance of NOMA. As a result, under certain ratios of the channel gains, grouping two users into the same TDMA time slot and serving them via NOMA may offer no advantage-and may even be detrimental-compared to serving them separately in distinct time slots. 

The contributions of the paper are summarized as follows.
\begin{itemize}
    \item For a given pair of two users, we determine the range of their channel gain ratio within which using NOMA yields a better sum-rate performance than using TDMA. Specifically, two optimization problems are formulated and solved using the Successive Convex Approximation (SCA) method to derive the upper and lower bounds of this range.   
    \item Based on the identified range, we propose a grouping algorithm to decide which users should be paired using NOMA and which should be served individually in one TDMA time slot. 
\end{itemize}

\section{System model}

\begin{figure}[t]
    \centering
    \includegraphics[width=0.35\textwidth, height=0.28\textwidth]{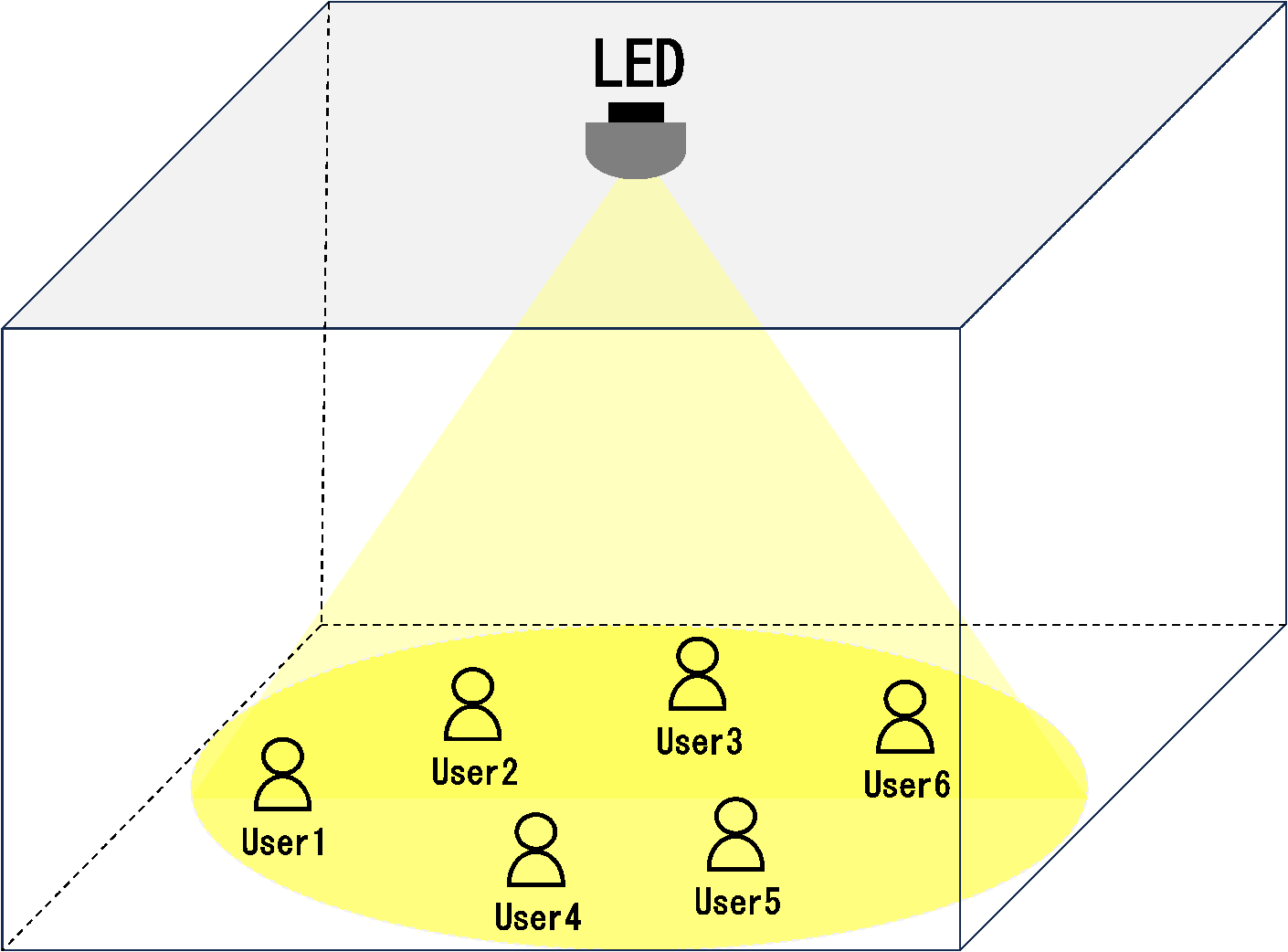}
    \caption{Multi-user VLC system model.}
    \label{fig:system}
\end{figure}
We consider a multi-user downlink VLC system, as depicted in Fig.~\ref{fig:system}, which consists of a single LED-based transmitter. This system serves $K$ users ($K = 6$ as illustrated in Fig.~\ref{fig:system}) using a hybrid NOMA-TDMA scheme. The LED transmitter is centrally positioned on the ceiling, while the users are randomly distributed throughout the room. As exampled in Fig.~\ref{fig:NOMA-TDMA}, our proposed scheme categorizes users into several groups, consisting of one or two users, with each group being allocated a TDMA time slot. Two users are served in one time slot using NOMA (i.e., users 1 and 5 in the first time slot and users 2 and 6 in the second time slot) if their channel gain ratio falls within a specific range over which using NOMA is more beneficial.   
\begin{figure}[ht]
    \centering
    \includegraphics[width=0.48\textwidth, height=0.28\textwidth]{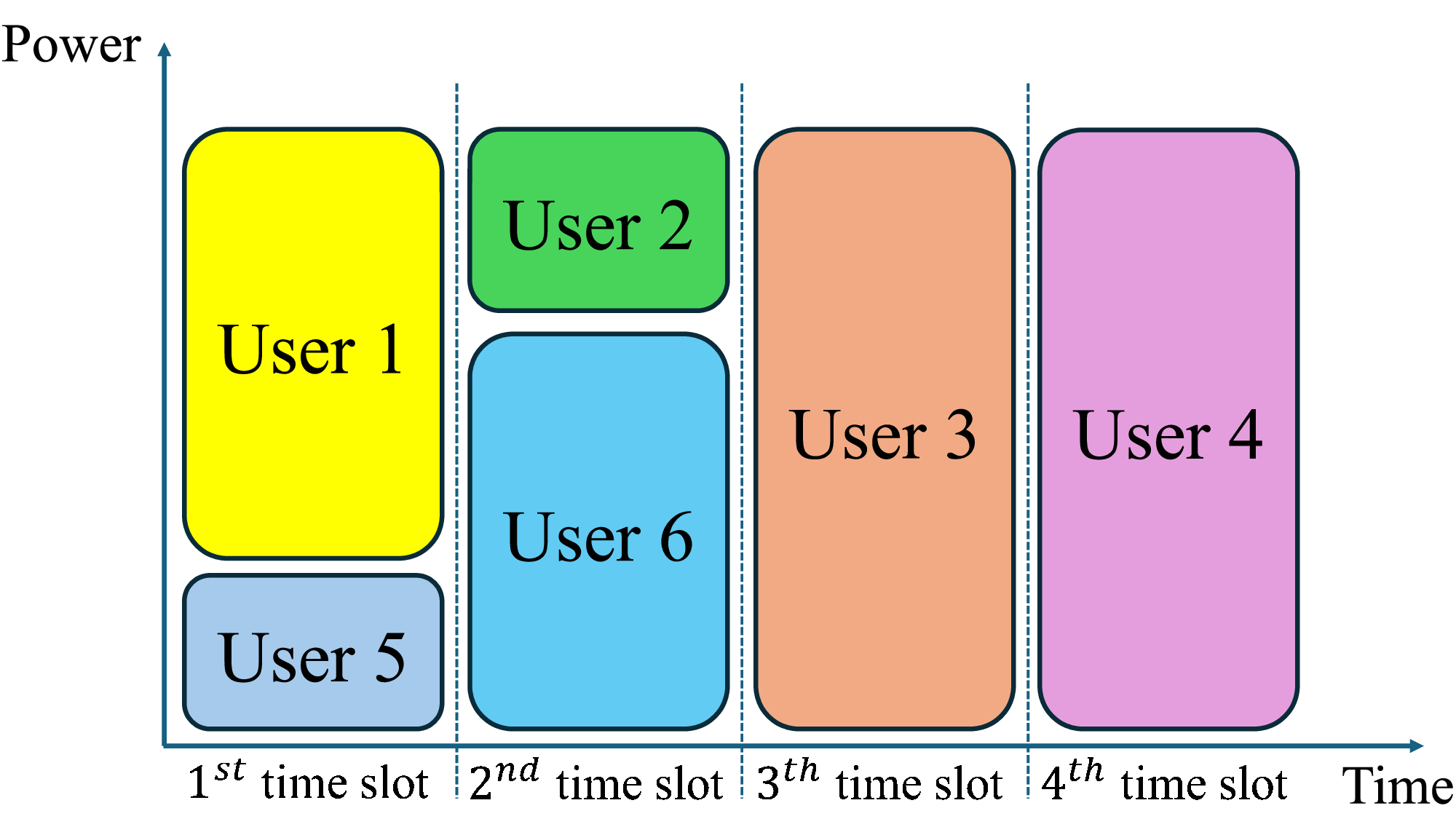}
    \caption{Proposed adaptive hybrid NOMA-TDMA scheme.}
    \label{fig:NOMA-TDMA}
\end{figure}

Assume that in the $k$-th time slot, NOMA is employed to serve a two-user group. The transmitted optical signal can be expressed as
\begin{equation}
x_{k} = \sqrt{ \alpha_{k,1} P_{\text{LED}}} s_{k,1} + \sqrt{ \alpha_{k,2} P_{\text{LED}}} s_{k,2} + P_{\text{D}},
\label{eq:1}
\end{equation}
where $P_{\text{LED}}$ is the transmitted power of the LED, $s_{k,1}$ and $s_{k,2}$ are the transmitted data of the two users, and $P_{\text{D}}$ is the power offset which maintains the non-negativity of the optical signal and controls the brightness of the LED. In \eqref{eq:1}, 
$\alpha_{k,1}$ and $\alpha_{k,2}$ are power allocation coefficients that determine how transmission power is distributed between 2 users. Since determining the optimal $\alpha_{k,1}$ and $\alpha_{k,2}$ is an NP-hard problem, we employ the Fractional Transmit Power Allocation (FTPA) scheme to set values for these coefficients \cite{tao2019strategy}
\begin{align}
\alpha_{k,1} = \frac{h_{k,2}^{2\alpha}} {h_{k,1}^{2\alpha}+h_{k,2}^{2\alpha}},~ 
\alpha_{k,2} = \frac{h_{k,1}^{2\alpha}} {h_{k,1}^{2\alpha}+h_{k,2}^{2\alpha}} ,
\label{eq:2}
\end{align}
where $h_{k,1}$ and $h_{k,2}$ represent the channel gains of the two users in the pair. In this study, $\alpha = 1$ is chosen for simplicity. 

The electrical signal received by the user $i$ $(i \in \{1,2\})$ in the $k$-th time slot is written by 
\begin{equation}
y_{k,i} = \xi h_{k,i} x_{k} + n_{k,i},
\label{eq:3}
\end{equation}
where $\xi$ is the optical-electrical conversion efficiency, $n_{k,i}$ is an additive white Gaussian noise (AWGN) with zero-mean and variance $\sigma^2_{k,i}$, which primarily originates from shot noise and thermal noise, and $h_{k}$ is the channel gain of the VLC link. In this work, we consider the line-of-sight (LoS) link component, which can be represented as \cite{komine2004fundamental}
\begin{align}
h_{k,i}  = & \frac {(m+1) A_{P} R_{P}}{2 \pi d_{k,i}^2} \cos^m(\phi_{k,i}) T_{s}(\varphi_{k,i}) T_{f}(\varphi_{k,i}) \cos (\varphi_{k,i}) \nonumber \\
& \times \mathds{1}_{[0, \Psi]}(\phi_{k, i}),
\label{eq:4}
\end{align}
where $\mathds{1}_{[x, y]}(\cdot)$ denotes the indicator function. In \eqref{eq:4}, $m = -1/\log_{2}(\cos(\phi_{1/2}))$, with $\phi_{1/2}$ being the semi-angle at half illuminance of the LED, is the Lambertian radiation, $\phi_{k,i}$ and $\varphi_{k,i}$ denote the angles of irradiance and incidence between the LED and the user, $\Psi$ denotes the photodiode’s field of view (FOV), $d_{k,i}$ is the distance between the LED and the user, $A_{P}$ and $R_{P}$ represent the detection area and responsivity of the photodiode, respectively. Also, $T_{s}(\varphi_{k,i})$ represents the gain of the optical filter, and $T_{f}(\varphi_{k,i})$ represents the non-imaging concentrator gain, which is given by
\begin{align}
    T_{f}(\varphi_{k,i}) = \frac {\kappa^2}{\sin^2(\varphi_{k,i})}\mathds{1}_{[0, \Psi]}(\varphi_{k, i}),
\end{align}
where $\kappa$ is the refractive index of the concentrator.

\section{Proposed Adaptive Hybrid NOMA-TDMA}
\subsection{Channel condition for using NOMA}
Given a pair of users, the proposed scheme first determines whether they should be grouped into one time slot and served using NOMA. For this purpose, we identify the range of channel gain ratios within which using NOMA yields a higher sum-rate performance than using TDMA. 

Firstly, let us consider the case that two users $u_{k, 1}$ and $u_{k, 2}$ (with their respective channel gains are $h_{k, 1}$ and $h_{k, 2}$) are grouped into the $k$-th time slot and served by NOMA. Assume that $u_{k, 1}$ and $u_{k, 2}$ are weak and strong users, respectively (i.e., $h_{k,1} \leq h_{k,2}$). For the user $u_{k, 1}$, it treats $s_{k, 2}$ as noise and directly decodes $s_{k, 1}$. For the strong user $u_{k, 2}$ employing SIC, it first decodes $s_{k, 1}$, treating $s_{k, 2}$ as noise. After decoding $s_{k, 1}$, it subtracts $s_{k, 1}$ from the received signal and then decodes $s_{k, 2}$. Denote $\tau_k$ as the duration of the $k$-th time slot. The achievable sum-rate of the two users is given by \cite{wang2013tight}
\begin{align}
R_{k, \text{NOMA}} & = \tau_{k} \log_2\left( 1 + \frac {e} {2\pi} \frac {\alpha_{k,1} P_{\text{LED}} h^2_{k,1}} {\alpha_{k,2} P_{\text{LED}} h^2_{k,1}+\sigma_{k,1}^2}\right) \nonumber \\
& + \tau_{k} \log_2\left( 1 + \frac {e} {2\pi} \frac {\alpha_{k,2} P_{\text{LED}} h^2_{k,2}} {\sigma_{k,2}^2}\right), 
\label{eq:6}
\end{align}
where $\sigma_{k,1}^2$ and $\sigma_{k,2}^2$ are the power of the noise in the signals received by $u_{k, 1}$ and $u_{k, 2}$, respectively. Now, in the case that TDMA is used, the achievable sum-rate is 
\begin{align}
R_{k, \text{TDMA}} & = {\tau_{k,1}} \log_2\left( 1 + \frac {e} {2\pi} \frac {P_{\text{LED}} h^2_{k,1}} {\sigma_{k,1}^2}\right) \nonumber \\
& + {\tau_{k,2}} \log_2\left( 1 + \frac {e} {2\pi} \frac {P_{\text{LED}} h^2_{k,2}} {\sigma_{k,2}^2}\right), 
\label{eq:7}
\end{align}
where $\tau_{k, 1}$ and $\tau_{k, 2}$ denote the allocated time slots to $u_{k, 1}$ and $u_{k, 2}$, satisfying $\tau_{k, 1} + \tau_{k, 2} = \tau_k$. For simplicity, let $\tau_{k, 1} = \tau_{k, 2} = \frac{\tau_k}{2}$ and suppose that $\sigma_{k,1}^2 = \sigma_{k,2}^2 = \sigma_k^2$.

Denote $\gamma_k = \frac{P_{\text{LED}} h^2_{k,1}}{\sigma_{k}^2}$ as the signal-to-noise ratio (SNR) of the weak user, and define $r = \frac{h_{k,2}^2}{h_{k,1}^2}$ as the squared ratio of the channel gains between the strong user and the weak user. Note that both $\gamma_k$ and $r$ are positive. Accordingly, the NOMA power allocation coefficients can be determined by the ratio $r$, such that $\alpha_{k,1} = \frac{r}{r+1}$ and $\alpha_{k,2} = \frac{1}{r+1}$. As this ratio directly affects the performance of the NOMA, we aim to determine its values at which using NOMA is more beneficial than using TDMA, i.e., $R_{k, \text{NOMA}} \geq R_{k, \text{TDMA}}$. In other words, identifying the values of $r$ that satisfy the following inequality
\begin{align}
& \log_2\left( 1 + \frac {e} {2\pi}  \frac{r \gamma_k}{r+\gamma_k+1} \right) +\log_2\left( 1 + \frac {e} {2\pi} \frac {r \gamma_k} {r+1}\right) \geq \nonumber \\
& \frac{1}{2}\log_2 \left( 1 + \frac {e} {2\pi} \gamma_k\right) + \frac{1}{2}\log_2 \left( 1 + \frac {e} {2\pi} r\gamma_k\right).
\label{eq:8}
\end{align}
In the following proposition, we show that there exists $r_{\text{min}}$ and $r_{\text{max}}$ ($r_{\text{min}}$ and $r_{\text{max}} > 0$) that the above inequality satisfies for $\forall r \in [r_{\text{min}}, r_{\text{max}}] $.

\newtheorem{proposition}{Proposition}
\begin{proposition}
In \eqref{eq:8}, let $p(r) = \log_2\left( 1 + \frac {e} {2\pi}  \frac{r \gamma_k}{r+\gamma_k+1} \right) +\log_2\left( 1 + \frac {e} {2\pi} \frac {r \gamma_k} {r+1}\right)$ and $q(r) = \frac{1}{2}\log_2 \left( 1 + \frac {e} {2\pi} \gamma_k\right) + \frac{1}{2}\log_2 \left( 1 + \frac {e} {2\pi} r\gamma_k\right)$. Then the equation $f(r) = p(r)-q(r)=0$ has at most 2 positive solutions $r_{k, \rm{min}}$, $r_{k, \rm{max}}$  and $f(r) \geq 0~ \forall r \in [r_{\rm{min}},r_{\rm{max}}]$.
\end{proposition}

\begin{proof}
The proof is given in the Appendix.
\end{proof}
\textbf{Proposition 1} shows that if $r \in [r_{\text{min}}, r_{\text{max}}]$, the system groups $u_{k, 1}$ and $u_{k, 2}$ into the $k$-th time slot and utilizes NOMA. Otherwise, the two users are allocated distinct TDMA time slots. Note that since $r_{\text{min}}$ and $r_{\text{max}}$ are the minimum and maximum values that satisfy the inequality in \eqref{eq:8}, they can be computed by solving the following optimization problems 

\begin{subequations}
\label{OptProbMin}
    \begin{alignat}{2}
        &\underset {} {\text{minimize}} & & r \label{eq:9}\\
        &\text{subject to }  & &  \nonumber \\
        & & & r \ge 0 , \label{eq:10} \\
        & & & p(r) - q(r) \ge 0 ,\label{eq:11}
    \end{alignat}
\end{subequations}
and
\begin{subequations}
\label{OptProbMax}
    \begin{alignat}{2}
        &\underset {} {\text{maximize}} & & r \label{eq:12}\\
        &\text{subject to }  & &  \nonumber \\
        & & & r \ge 0 \label{eq:13} ,\\
        & & & p(r) - q(r) \ge 0 \label{eq:14} .
    \end{alignat}
\end{subequations}
It can be easily shown that both $p(r)$ and $g(r)$ are concave functions of $r$. The constraints in \eqref{eq:11} and \eqref{eq:14} are, therefore, the difference of concave functions and generally not convex.
To address the non-convexity of these constraints, we employ the Successive Convex Approximation (SCA) technique \cite{razaviyayn2014successive}, which is an efficient iterative method for solving non-convex optimization problems.
Specifically, at the $i$-th iteration,  the function $q(r)$ is approximately linearized around the solution at the previous iteration by its first-order approximation. Replacing $q(r)$ by its linear approximation results in surrogate convex problems that can be efficiently solved using standard optimization solvers (i.e., SDPT3 and SeDuMi). The obtained solution is then used to refine the approximation in the next iteration. This iterative process continues until a convergence is achieved, typically determined by the change between successive solutions falling below a predefined threshold.
The detailed procedure is described in \textbf{Algorithm}~\ref{alg1}. 
\begin{algorithm}[ht]
\caption{SCA method to solve \eqref{OptProbMin} and \eqref{OptProbMax}}
\label{alg1}
\begin{algorithmic}[1]
\STATE \textbf{Initialize:} Choose feasible initial points $r^{(0)}_{\text{min}} > 0$ and $r^{(0)}_{\text{max}} > 0$, and a tolerance $\epsilon > 0$.
\REPEAT
    \STATE Compute the first-order approximations of $q(r)$
    \[
    \tilde{q}_{\text{min}}(r) = q\left(r^{(i-1)}_{\text{min}}\right) + q'\left(r^{(i-1)}_{\text{min}}\right)\left(r_{\text{min}} - r^{(i-1)}_{\text{min}}\right)
    \]
    \[
    \tilde{q}_{\text{max}}(r) = q\left(r^{(i-1)}_{\text{max}}\right) + q'\left(r^{(i-1)}_{\text{max}}\right)\left(r_{\text{max}} - r^{(i-1)}_{\text{max}}\right)
    \]
    where $r^{(i-1)}_{\text{min}}$ and $r^{(i-1)}_{\text{max}}$ are the solutions at the $(i-1)$-th iteration. 
    \STATE Replace the constraints \eqref{eq:11} and \eqref{eq:14} with $p(r) - \tilde{q}_{\text{min}}(r) \geq 0$ and $p(r) - \tilde{q}_{\text{max}}(r) \geq 0$, respectively, to form surrogate convex problems.
    \STATE Solve the surrogate problems to obtain solutions $r^{(i)}_{\text{min}}$ and $r^{(i)}_{\text{max}}$. 
\UNTIL $\left|r^{(i)}_{\text{min}} - r^{(i-1)}_{\text{min}}\right| < \epsilon$ and $\left|r^{(i)}_{\text{max}} - r^{(i-1)}_{\text{max}}\right| < \epsilon$
\RETURN Optimal solutions $r_{\text{min}} \leftarrow r^{(i)}_{\text{min}}$ to \eqref{OptProbMin} and $r_{\text{max}} \leftarrow r^{(i)}_{\text{max}}$ to \eqref{OptProbMax}.
\end{algorithmic}
\end{algorithm}
\subsection{Adaptive User Pairing Algorithm}
In this section, we propose a heuristic algorithm to select user pairs suitable for NOMA transmission. 
Specifically, the proposed algorithm first employs the same user pairing strategy proposed in~\cite{zhu2018optimal}, which pairs the weakest user with the strongest one, the second weakest user with the second strongest one, and so on, to maximize the contrast in channel gains between paired users. Additionally, it introduces an adaptive scheme that only pairs whose channel gain ratios fall within the range \( [r_{\text{min}}, r_{\text{max}}] \) determined in the previous section are considered valid. The detailed procedure is presented in \textbf{Algorithm}~\ref{alg2}.
\begin{algorithm}[ht]
\caption{User Pairing for NOMA}
\label{alg2}
\begin{algorithmic}[1]
\STATE \textbf{Input:} Channel gains of \( K \) users assuming that \( h_1 \leq h_2 \leq \dots \leq h_K \)
\STATE Initialize an empty pairing set: \( \mathcal{S} \leftarrow \emptyset \)
\STATE Mark all users as unpaired
\FOR{weak user index \( i = 1 \) to \( K - 1 \)}
    \FOR{strong user index \( j = K \) down to \( i + 1 \)}
        \IF{users \( i \) and \( j \) are unpaired \AND \( \frac{h_j^2}{h_i^2} \in [r_{\text{min}}, r_{\text{max}}] \)}
            \STATE Add the pair \( (i, j) \) to \( \mathcal{S} \)
            \STATE Mark users \( i \) and \( j \) as paired
            \STATE \textbf{Break} inner loop and move to the next weak user
        \ENDIF
    \ENDFOR
\ENDFOR
\RETURN the set of NOMA user pairs \( \mathcal{S} \)
\end{algorithmic}
\end{algorithm}
\textbf{Algorithm}~\ref{alg2} performs user pairing by examining combinations between the weakest and strongest unpaired users. Starting from the user with the weakest channel gain, it attempts to find a valid partner by checking users with stronger channel gains in reverse order, beginning with the strongest user and moving inward. A user pair is formed if their squared channel gain ratio \( \frac{h_j^2}{h_i^2} \) lies within the range \( [r_{\text{min}}, r_{\text{max}}] \) determined in \textbf{Algorithm}~\ref{alg1}. Once a valid pair is identified, both users are marked as paired and excluded from further consideration. The process then moves to the next unpaired weak user and repeats the search. Since the users are sorted in ascending order of channel gain, this strategy prioritizes pairing users with the most contrasting channel conditions, which is beneficial for NOMA performance. If no suitable strong user is found for a given weak user, the algorithm skips to the next weak user. Users who cannot be successfully paired are assumed to be served individually using TDMA.

\section{Simulation Results and Discussions}
This section presents numerical results to illustrate the performance of the proposed hybrid NOMA-TDMA scheme in comparison with TDMA and the hybrid method proposed in \cite{zhu2018optimal}. Without otherwise noted, the simulation parameters are given in Table \ref{tab: System parameters}. A 3D Cartesian coordinate system whose origin is the center of the floor is used to specify the positions of the LED and users.  
\begin{table}[ht]
    \centering
    \caption{Simulation parameters.}
    \begin{tabular}{l l}
\hline\hline
     \textbf{Parameter} & \textbf{Value}\\\hline\hline
     Room Dimension     & 6 m $\times$ 6 m $\times$ 3 m  \\\hline
      LED power, $P_{\text{LED}}$ & 1 W \\ \hline
      LED semi-angle at half illuminance, $\phi_{1/2}$    & $60^\circ$ \\\hline
      Energy conversion efficiency, $\eta$     & 0.44  \\\hline
      PD active area, $A_R$     & 1 $\text{cm}^2$  \\\hline
      PD responsivity, $R_P$     & 0.54 A/W  \\\hline
      PD field of view (FOV), $\Psi$     & $60^\circ $  \\\hline
     Optical filter gain, $T_s(\varphi_k)$      & 1  \\\hline
     Refractive index of concentrator, $\kappa$ &1.5  \\\hline
     Noise power, $\sigma_{k, 1}^2$, $\sigma_{k, 2}^2$      & $10^{-14}$ W \\\hline
    \end{tabular}
    \label{tab: System parameters}
\end{table} 
\begin{figure}[ht]
    \centering
    \includegraphics[width=0.4\textwidth, height = 0.3\textwidth]{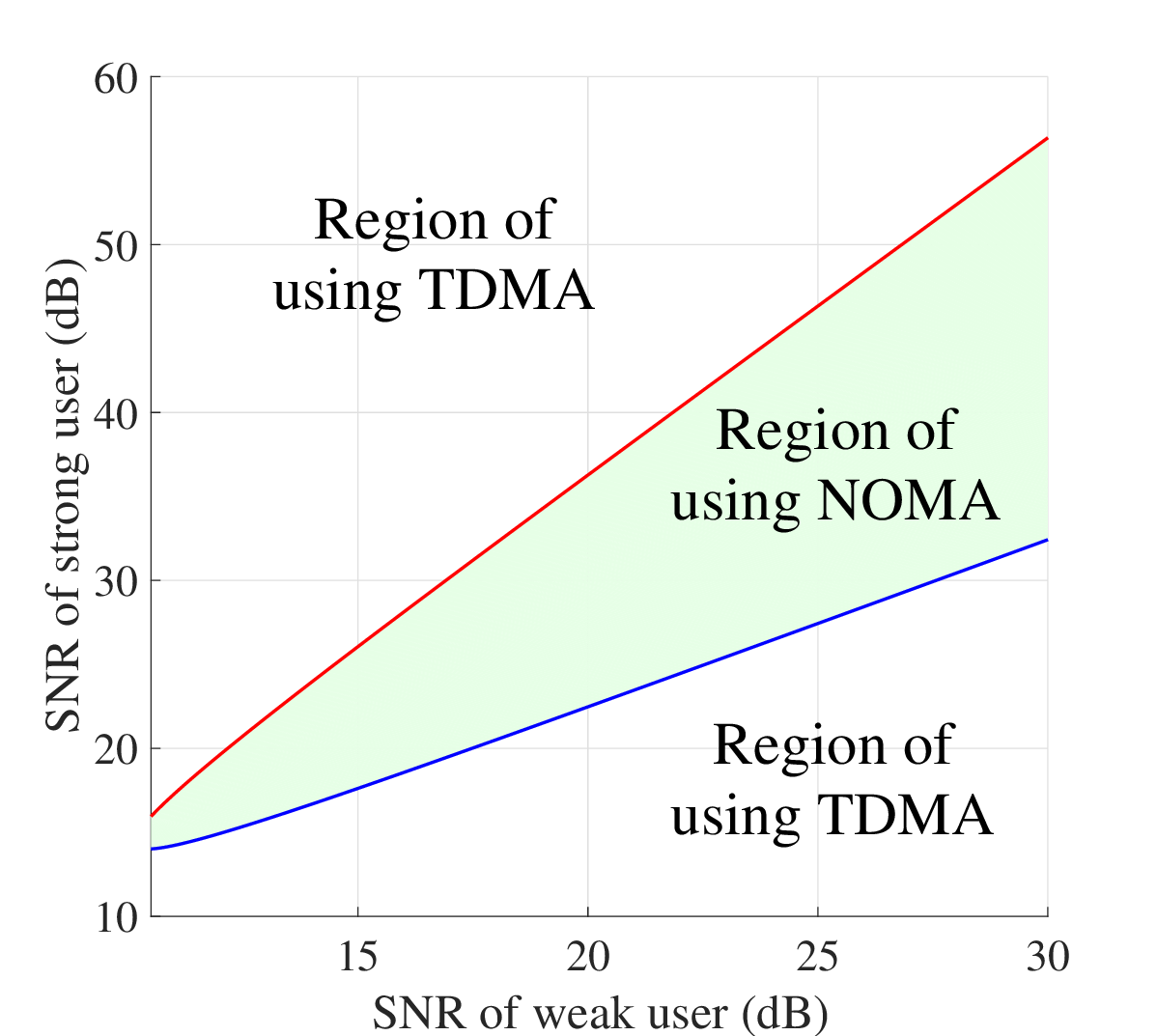}
    \caption{Region of using NOMA.}
    \label{fig:regions}
\end{figure}

From the resulting squared channel gain ratio bounds \( r_{\text{min}} \) and \( r_{\text{max}} \) determined in \textbf{Algorithm}~\ref{alg1}, we illustrate in Fig.~\ref{fig:regions} the SNR regions of the strong and weak users, which decide whether NOMA or TDMA should be employed. The shaded region represents the region where NOMA is expected to provide superior performance compared to TDMA. It is observed that, as the SNR of the weak user increases, the region of using NOMA gradually expands. This behavior indicates that under improved channel conditions for the weak user, the negative impact of similar channel gains on NOMA performance diminishes.

\begin{figure}[ht]
    \centering
    \includegraphics[width=0.4\textwidth, height = 0.3\textwidth]{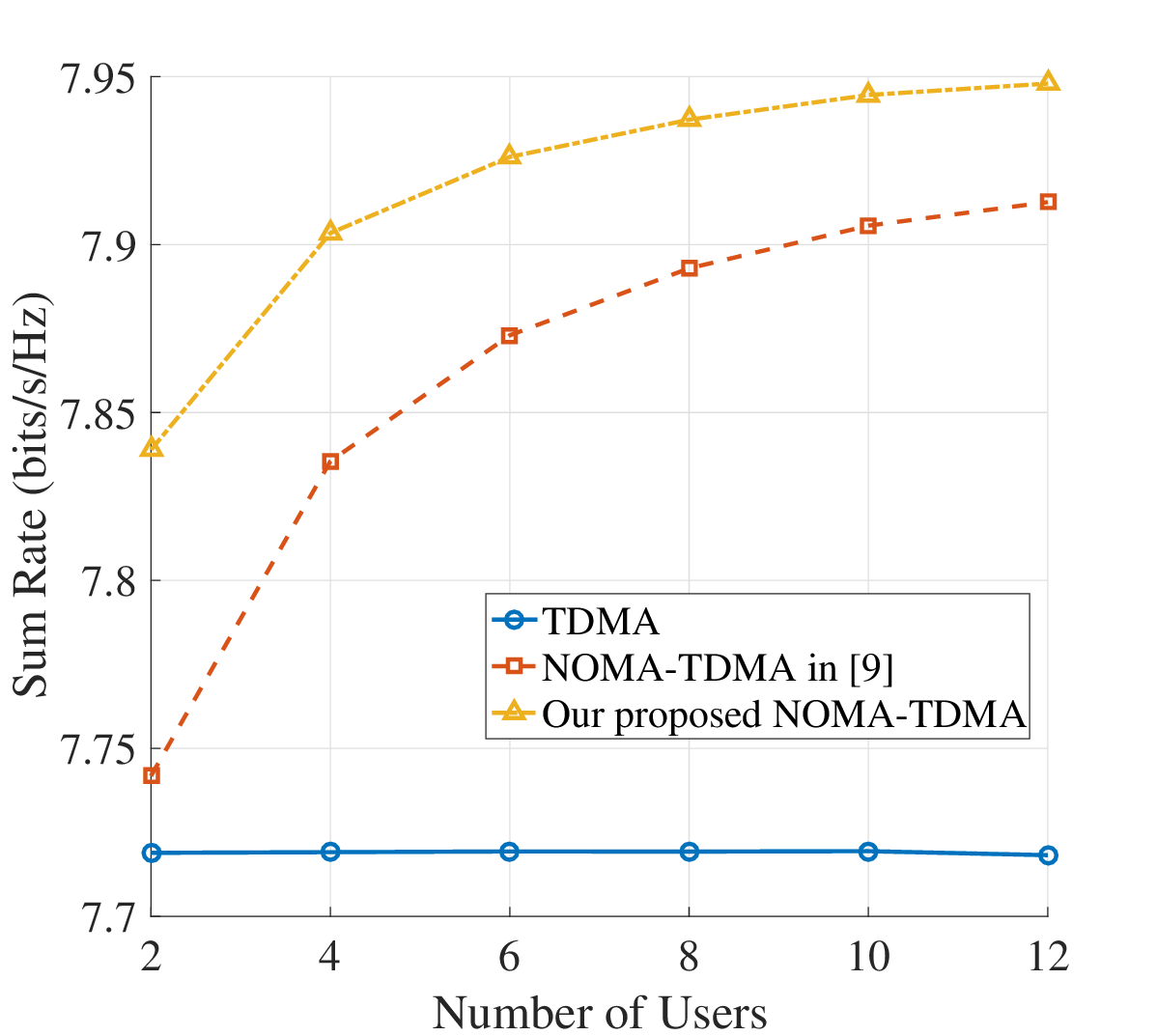}
    \caption{Sum rate for different multiple access schemes.}
    \label{fig:4}
\end{figure} 
 
In Fig.~\ref{fig:4}, we compare the average sum-rate performance of three different schemes: conventional TDMA, NOMA-TDMA pairing strategy in \cite{zhu2018optimal}, and the proposed adaptive NOMA-TDMA.  As shown in the figure, our proposed NOMA-TDMA scheme that employs an adaptive pairing achieves the highest average sum-rate regardless of the number of users.
This demonstrates the effectiveness of incorporating the adaptive pair scheme, which ensures that only users with sufficiently distinct channel conditions are grouped for NOMA transmission. 

Furthermore, Fig.~\ref{fig:P_LED} demonstrates the robustness of our proposed pairing strategy, particularly in scenarios where users have similar channel gains. For example, in a scenario involving six users located at (2.5, 5.5, 0), (4, 0, 0), (5, 1, 0), (5, 5.5, 0), (5, 6, 0), and (6, 1, 0), the sum-rate performance of our scheme consistently surpasses that of the method in \cite{zhu2018optimal} by approximately 0.25~bits/s/Hz, regardless of the LED power.
\begin{figure}[ht]
    \centering
    \includegraphics[width=0.4\textwidth, height = 0.3\textwidth]{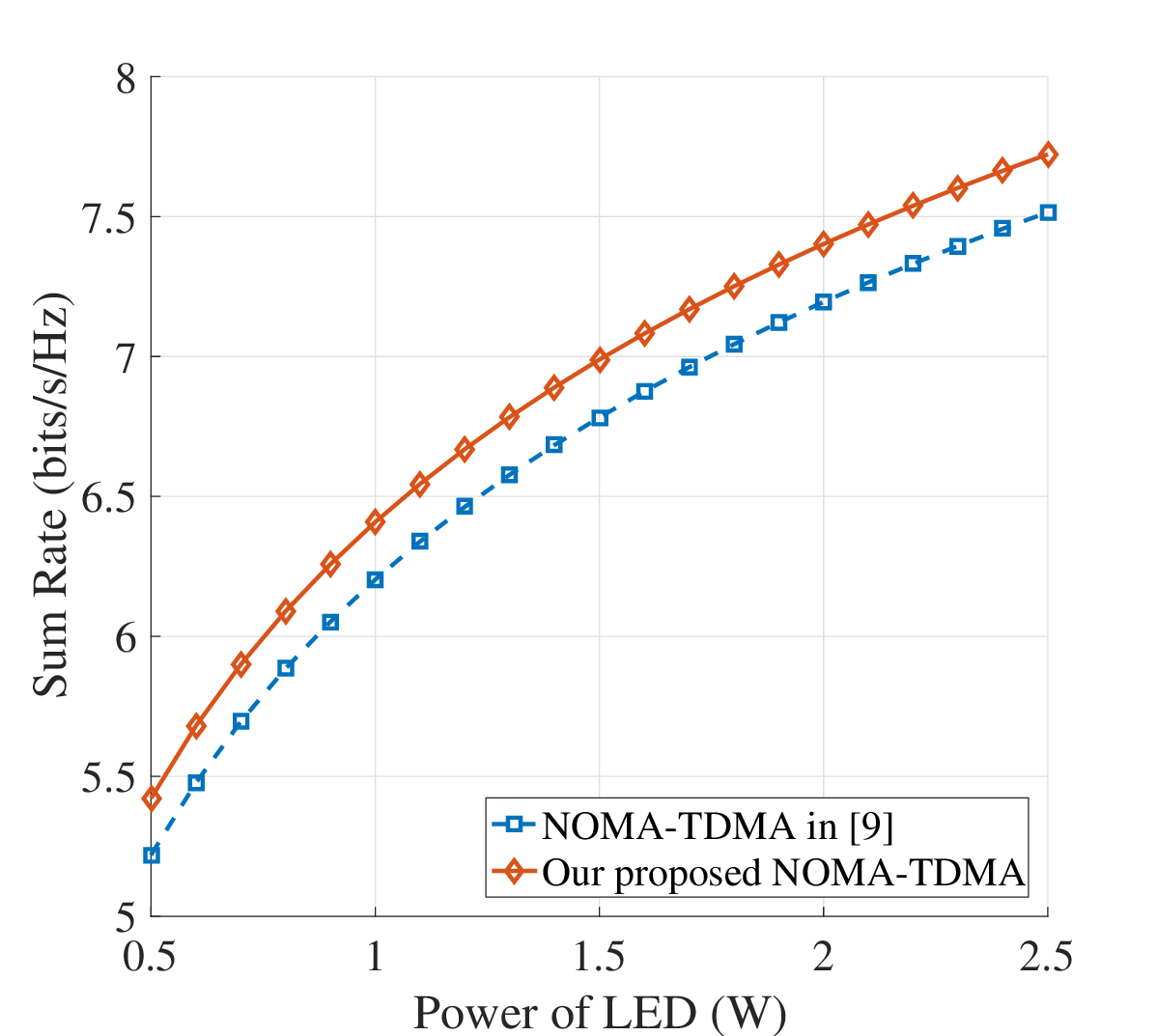}
    \caption{Sum rate for different multiple access schemes.}
    \label{fig:P_LED}
\end{figure} 

\section{Conclusion}
In this paper, we propose an adaptive hybrid NOMA-TDMA transmission scheme for multi-user VLC networks. The proposed scheme pairs two users into one TDMA time slot if their squared channel gain ratio is within a certain range, where NOMA performs better than TDMA. 
To determine the applicable region for NOMA, two optimization problems are formulated and solved using the SCA method. Numerical results showed that the proposed hybrid scheme with an adaptive pairing outperforms the traditional hybrid NOMA-TDMA method. For future work, we will consider the impact of imperfect SIC  to evaluate the robustness of the proposed scheme under more practical conditions.

\appendix
To prove that $ f(r) = p(r) - q(r) = 0 $ has at most 2 solutions $r_{\text{min}}$, $r_{\text{max}} \in [0,\infty)$, and $f(r) \geq 0~\forall r \in [r_{\text{min}}, r_{\text{max}}]$, we will show that the equation \( f'(r) = 0 \) with $f'(r)$ being the derivative of $f(r)$ has at most one solution. We have
\begin{align}
f'(r) &= \frac{1}{\ln 2} \left( 
\frac{t\gamma_k}{(1 + r + t\gamma_kr)(1 + r)} \right. \nonumber \\
&\quad \left. + \frac{t\gamma_k(1 + t\gamma_k)}{(1 + r + \gamma_k + t\gamma_k r)(1 + r + \gamma_k)} 
- \frac{0.5t\gamma_k}{1 + t\gamma_kr} \right),
\label{eq:5}
\end{align}
where $t = \frac{e}{2\pi}$. 

The expression in \eqref{eq:5} can be written in a fractional form $f'(r) = \frac{u(r)}{v(r)}$ where the numerator is a quartic polynomial of the form $u(r) =  f_1(\gamma_k)r^4 + f_2(\gamma_k)r^3 + f_3(\gamma_k)r^2 + f_4(\gamma_k)r + f_5(\gamma_k)$, with
\begin{align}
    f_1(\gamma_k) = -0.5(t\gamma_k+1)^2
\end{align}
\begin{align}
    f_2(\gamma_k) = (t\gamma_k+1)((t^2-0.5t)\gamma_k^2+(t-1)\gamma_k-2)
\end{align}
\begin{align}
    f_3(\gamma_k) =& ~ (t^3+0.5t^2-0.5t)\gamma_k^3 +(4.5t^2-t-0.5)\gamma_k^2 \nonumber \\ & + (4t-3)\gamma_k -1
\end{align}
\begin{align}
    f_4(\gamma_k) = 0.5t\gamma_k^3 + (2t^2 + 1.5t - 1)\gamma_k^2 + (5t - 1)\gamma_k + 2
\end{align}
\begin{align}
    f_5(\gamma_k) = 0.5(\gamma_k+1)^2+t\gamma_k
\end{align}
Since \( \gamma_k > 0 \), it follows that  \( f_1(\gamma_k) < 0 \), \( f_2(\gamma_k) < 0 \), \( f_3(\gamma_k) < 0 \), \( f_4(\gamma_k) > 0 \), and \( f_5(\gamma_k) > 0 \). 
Combining with the observations that \( u(0) > 0 \), \( \underset{r \to \infty}{\lim} u(r) < 0 \), and \( u''(r) < 0 \) \(\forall r \in [0, \infty) \), it can be deduced that \( u(r) \) has exactly one root in this domain. Therefore, \( f(r) \) can have one critical point (local maximum), which implies that it can cross zero at most twice at \(r_{min}\) and \(r_{max}\). Since \( f(r) \) has exactly one local maximum, it follows that \( f(r) \geq 0 \) throughout the interval \( [r_{\min}, r_{\max}] \).
\bibliographystyle{ieeetr}
\bibliography{references}

\end{document}